\documentclass[amsa]{ipart}

\RequirePackage{hyperref}
\usepackage{framed}
\usepackage{graphicx}
\usepackage{epstopdf}
\usepackage{subfigure}
\usepackage{epsfig}
\usepackage{multirow}
\usepackage{booktabs}

\DeclareMathOperator*{\argmin}{arg\,min}

\DeclareMathOperator*{\diag}{Diag}
\DeclareMathOperator*{\dom}{dom}

\def\A{{\mathcal A}}
\def\B{{\mathcal B}}
\def\C{{\mathcal C}}
\def\D{{\mathcal D}}

\def\L{{\mathcal L}}

\def\O{{\mathcal O}}
\def\P{{\mathcal P}}

\def\X{{\mathcal X}}
\def\Y{{\mathcal Y}}

\startlocaldefs
\theoremstyle{plain}

\newtheorem{theorem}{Theorem}[section]
\newtheorem{assumption}{Assumption}[section]

\endlocaldefs

\pubyear{2022}
\volume{0}
\issue{0}
\firstpage{1}
\lastpage{9}

\begin{document}

\begin{frontmatter}
\title[SSNAL]{Semismooth Newton Augmented Lagrangian Algorithm for Adaptive Lasso Penalized Least Squares in Semiparametric Regression}

\begin{aug}
    \author{\fnms{Peili Li}\ead[label=e1]{lipeili@henu.edu.cn}},
    \address{
     School of Mathematics and Statistics, Henan University, Kaifeng 475000\\
     Center for Applied Mathematics of Henan Province, Henan University, Zhengzhou 450046\\
             China\\
             \printead{e1}}
    \author{\fnms{Yunhai Xiao}\thanksref{t2}\ead[label=e2]{yhxiao@henu.edu.cn}},
    \address{School of Mathematics and Statistics, Henan University, Kaifeng 475000\\
     Center for Applied Mathematics of Henan Province, Henan University, Zhengzhou 450046\\
             China\\
             \printead{e2}}
    \author{\fnms{Meixia Yang} \ead[label=e3]{mathmxyang@163.com}}
    \address{School of Mathematics and Statistics, Henan University, Kaifeng 475000\\
             China\\
             \printead{e3}}
    \and
        \author{\fnms{Hanbing Zhu}\ead[label=e4]{zhuhbecnu@163.com}}
    \address{School of Big Data and Statistics, Anhui University, Hefei 230601\\
             China\\
             \printead{e4}}
    \thankstext{t2}{Corresponding author}
\end{aug}

\begin{abstract}
This paper is concerned with a partially linear semiparametric regression model containing an unknown regression coefficient, an unknown nonparametric function, and an unobservable Gaussian distributed random error. We focus on the case of simultaneous variable selection and estimation with a divergent number of covariates under the assumption that the regression coefficient is sparse. We consider the applications of the least squares to semiparametric regression and particularly present an adaptive lasso penalized least squares (PLS) method to select the regression coefficient. We note that there are many algorithms for PLS  in various applications, but they seem to be rarely used in semiparametric regression. This paper focuses on using a semismooth Newton augmented Lagrangian (SSNAL) algorithm to solve the dual of PLS which is the sum of a smooth strongly convex function and an indicator function. At each iteration, there must be a strongly semismooth nonlinear system,  which can be solved by semismooth Newton by making full use of the penalized term. We show that the algorithm offers a significant computational advantage, and the semismooth Newton method admits fast local convergence rate. Numerical experiments on simulated and real data have demonstrated the effectiveness of the PLS method and the progressiveness of the SSNAL algorithm.
\end{abstract}



\end{frontmatter}

\section{Introduction}\label{section1}

Statistical inference on a multidimensional random variable commonly focuses on functionals of its distribution that are either purely parametric, or purely nonparametric, or semiparametric as an intermediate strategy.
Semiparametric regression makes full use of the known information, makes up for the shortcomings of nonparametric, and gives full play to the advantages of the parametric.
Suppose that the random sample $\{(X_{i},T_{i},Y_{i})\}_{i=1}^n$ is generated from the following partially linear semiparametric regression model
\begin{equation}\label{sprm}
Y_{i}=X_{i}^{\top}\beta+g(T_{i})+\varepsilon_{i},\quad i=1,2,\ldots,n,
\end{equation}
where $Y_{i}$'s are scalar response variates, $X_i$'s are $p$-variate covariates, $T_i$'s are $d$-variate covariates,
$(X_i, T_i)$ are either independent and identically distributed (i.i.d.) random design points or fixed design points;
$\beta$ is an unknown $p$-variate regression coefficient, $g(\cdot)$ is an unknown measurable function from $\mathbb{R}^d$ to $\mathbb{R}$, and $\varepsilon_i$'s are random statistical errors. It is assumed that the errors $\varepsilon_i$'s are i.i.d. random variates and independent of $\{(X_{i},T_{i})\}$ with zero mean and variance $\sigma^2$.
Without loss of generality, we assume that $T_{i}$ and $g(\cdot)$ are scaled into the closed interval $[0,1]$.
Given the data $\{(X_{i},T_{i},Y_{i})\}_{i=1}^n$, the aim of partly linear semiparametric regression is to estimate the coefficient $\beta$ and the function $g(\cdot)$ from the data.

The interest in semiparametric regression model has grown significantly over the past few decades since it was introduced by \cite{Engle}  to analyze the relationship between temperature and electricity usage.
The model has been widely studied in a large variety of fields, such as finance, economics, geology and biology, to name only a few.
For an excellent survey, one can refer to the book of \cite{BOOK1}.
A potential challenge of estimation in this model is that it is composed of a finite-dimensional coefficient  $\beta$, and an infinite-dimensional parameter $g(\cdot)$.
We know that Least squares (LS) method is effective to find the optimal estimation of unknown quantity from an error contained observation \citep{Stigler}.
It is linear, unbiased and minimum variance, and particularly based on the famous Gauss-Markov theorem in linearly parametric regression model.

In recent years, high data dimensionality has brought unprecedented challenges and attracted increasing research attention.
When the data dimension diverges, variable selection through penalty functions is particularly effective, and selecting variables and estimating parameters are possible to be achieved simultaneously.
The commonly used penalty functions include lasso \cite{LASSO}, fused lasso \cite{FLASSO}, adaptive lasso \cite{Zou}, and SCAD \cite{FanLi}.
Combining these penalty functions with LS, various powerful penalization methods have been developed for variable selection in the literature.
For examples, \cite{long-data} employed the SCAD penalized least squares (PLS) for semiparametric model in longitudinal data analysis.
\cite{spline} applied the SCAD PLS to achieve sparsity in linear part and use polynomial spline to estimate the nonparametric part in partially linear model.
\cite{error} studied the SCAD PLS for partially linear models with measurement errors.
\cite{double} proposed a double-PLS method for partially linear model using the smoothing spline to estimate the nonparametric part and applying a shrinkage penalty on parametric components to achieve model parsimony.
\cite{K2021} applies an adaptive lasso penalty to both the linear and nonlinear components and uses a signed-rank technique for estimation.

In this paper, we are also interested in the PLS for parameter estimation and variable selection in the semiparametric regression model (\ref{sprm}) with diverging numbers of parameters.
The model (\ref{sprm}) can replace the baseline function by the estimator obtained under the assumption that the parameter is known, then it can be approximately regarded as a purely linear regression model.
Based on the linear regression model, we construct PLS function with adaptive lasso for regression parameter.
We show that its oracle properties can be proved in a similar way to the work of \cite{Zou}.
Seeing from a numerical point of view, the PLS is exactly a sum of a smooth strongly convex function and a nonsmooth adaptive lasso penalty term, so that it can be solve via various structured algorithms, such as the first-order alternating direction method of multipliers (ADMM) \citep{GM1976} and the second-order semismooth Newton method \citep{BYRD,LST}. Given the rare use of higher-order methods in semiparametric regression, this paper focuses on the second-order method to solve the PLS problem by making use of the second-order information of the adaptive lasso penalty.
We observe that the dual problem of PLS  consists of a smooth strongly convex function and an indicator function, which inspires us to employ the semismooth Newton augmented Lagrangian (SSNAL) method of \cite{LST} to solve it.
The most notable feature of this method is that there involves a strongly semismooth nonlinear system which comes from the proximal mapping of the adaptive lasso penalty. For this nonlinear systems, we note that its generalized Jacobian at its solution is symmetric and positive definite, so it is highly possible to design an efficient algorithm.
Finally, we conduct some numerical experiments on some synthetic data and real data sets.
The numerical results illustrate that PLS method with an adaptive lasso penalty is effective, and the employed SSNAL method is more advanced than ADMM.

The remaining parts of this paper are organized as follows.
In Section \ref{presult}, we quickly review some basic concepts in convex analysis and key ingredients needed for our subsequent developments.
In Section \ref{model}, we propose PLS with adaptive lasso for regression coefficient, and then construct its dual formulation and optimality condition.
In Section \ref{algorithm}, we use SSNAL to solve the dual problem and employ a semismooth Newton (SSN) to the involved semismooth nonlinear system.
In Section \ref{numer}, we report the numerical experiments by using some benchmark data.
Finally, we conclude our paper in Section \ref{consec}.

To end this section, we summarize some notations used in this paper.
For variates $x$, its $i$-th entry is denoted by $x_{i}$.
We denote $X=\text{Diag}(x)$ as a diagonal matrix with its $i$-th entry on the diagonal being $x_{i}$.
For variates $x$, we denote $\B_{\infty}^{(\tau)}:=\{x\ | \ \|x\|_{\infty}\leq\tau\}$, or $(\B_{\infty}^{(\tau)})_i:=\{x_{i}\ | \ |x_{i}|\leq\tau\}$ at component wise.
The $\ell_{1}$-norm (a.k.a. lasso), $\ell_2$-norm, and $\ell_{\infty}$-norm of a $p$-variates are defined by, respectively, $\|x\|_{1}:=\sum_i |x_{i}|$,  $\|x\|_2:=\sqrt{\sum_i x_{i}^{2}}$, and $\|x\|_{\infty}:=\max_i |x_{i}|$.
The transpose operation of a variates or a matrix is denoted by superscript ``$\top$".
For a linear operator $\A$, its adjoint is represented by $\A^{\ast}$, or $\A^{\top}$ at matrix case.
For variates $x$, $y$ with appropriate sizes, we define $\langle x,y\rangle=x^{\top}y$.
For a nonempty closed convex set $\C$, the symbol $\delta_{\C}(x)$ represents
an indicator function over $\C$ such that $\delta_{\C}(x)=0$ if $x\in\C$ and $+\infty$ otherwise.
We denote $I_{p}$ and $\mathbf{0}_{p}$ as $p$-dimensional identity matrix and zero vector, respectively.

\section{Preliminaries}\label{presult}

In this section, we summarize some basic concepts in convex analysis and briefly recall the SSNAL method for subsequent developments.
\subsection{Basic Concepts}\label{sub1}

Let $\X$ be finite dimensional real Euclidean space equipped with an inner product $\langle \cdot,\cdot \rangle $ and its induced norm $\|\cdot\|_2$.
For any $z\in\X$, the metric projection of $z$ onto $\C$ denoted by $\Pi_{\C}(z)$ is the optimal solution of the
minimization problem $\min\limits_y\{ \|y-z\|_2 \ | \ y\in\C\}$. Let $f:\X\rightarrow(-\infty,+\infty]$ be a closed proper convex function.
The effective domain of $f$ is defined by $\text{dom}(f):=\{x \ | \ f(x)<+\infty\}$.
The subdifferential of $f$ at $x \in \text{dom}(f)$ is defined by $\partial f(x):=\{x^*\ | \ f(z)\geq f(x)+\langle x^*,z-x\rangle, \ \forall z\in\X\}$.
Obviously, $\partial f(x)$ is a closed convex set when it is not empty \citep{RR}.

The dual norm $\|\cdot\|_*$ of a norm $\|\cdot\|$ is defined by:
$$
\|x\|_*:=\sup_{y\in\X}\{x^\top y \ | \ \|y\|\leq 1\}.
$$
It is easy to see that the dual norm of $\|\cdot\|_2$ is itself, and the $\ell_{1}$-norm and  $\ell_{\infty}$-norm are dual with respect to each other.
The Fenchel conjugate of a convex $f$ at $y\in\X$ is defined by
$$
f^\star(y):=\sup_{x\in \X}\{\langle x,y\rangle-f(x)\}=-\inf_{x\in \X}\{f(x)-\langle x,y\rangle\}, \quad \forall y\in \X.
$$
It is well known that the conjugate function $f^\star(y)$ is always convex and closed, proper if and only if $f$ is proper and convex \citep{RR}.
For any $x\in\X$, there exists a $y\in\X$ such that $y\in\partial f(x)$ or equivalently $x\in\partial f^{\star}(y)$ owing to a fact of $f$ being closed and convex \cite[Theorem 23.5]{RR}. Using the
definition of the dual norm, it is easy to deduce that the Fenchel conjugate of $\|x\|_{1}$ is  $\|x\|_{1}^\star=\delta_{\B^{(1)}_{\infty}}(x)$.

We have the following result for the metric projection (in $\ell_{\infty}$-norm) onto $\B^{(r)}_{\infty}$.
Given $x\in\X$, the orthogonal projection onto $\B^{(r)}_{\infty}$ is defined by
\begin{align*}
\Pi_{\B^{(r)}_{\infty}}(x)=\min\{r,\max\{x,-r\}\}.
\end{align*}
For any closed proper convex function $f:\X\rightarrow(-\infty,+\infty]$, the Moreau-Yosida regularization of $f$ at $x\in\X$ with positive
scalar $\tau>0$ is defined by
\begin{align*}
\varphi_{f}^{\tau}(x):=\min\limits_{y \in \X}\big\{f(y)+\frac{1}{2\tau}\|y-x\|^2_2\big\}.
\end{align*}
Moreover, the above problem has an unique optimal solution, which is known as the proximal mapping of $x$ associated with $f$, i.e.,
\begin{align*}
\P_{f}^{\tau}(x):=\argmin\limits_{y \in \X}\big\{f(y)+\frac{1}{2\tau}\|y-x\|^2_2\big\}.
\end{align*}
The proximal mapping of the $\ell_1$-norm function at point $x$ obeys the following form
$$
\P_{\|\cdot\|_1}^{\tau}(x)=\text{sgn} (x) \odot \max\big\{|x|-\tau,0\big\},
$$
where $\odot$ is Hadamard product, and the sign function $\text{sgn}(\cdot)$ and absolute value function $|\cdot|$ are component-wise. It is also well known that $\P_f^\tau(\cdot)$ is firmly non-expansive and globally Lipschitz continuous with modulus $1$.
For any $x\in\X$, the Moreau decomposition is expressed as $x=\P_{f}^{\tau}(x)+\tau\P_{f^\star}^{1/\tau}(x/\tau)$.
For an example, the proximal mapping of $\ell_1$-norm at $x$ can be expressed as $\P_{\|\cdot\|_1}^\tau(x)=x-\Pi_{\B_\infty^{(\tau)}}(x)$, which will be used frequently at the following parts.

\subsection{Review on SSNAL}\label{sub2}

Consider the general convex composite optimization model
\begin{equation}\label{ssnp}
\min_{x\in\X}\big\{f(x):=h(\A x)-\langle c,x\rangle+p(x)\big\},
\end{equation}
where $\A:\X\rightarrow\Y$ is a linear map, $h:\Y\rightarrow\mathbb{R}$ and $p:\X\rightarrow(-\infty,+\infty]$ are two closed proper convex functions, and $c\in\X$ is a given variates.
We assume that $h$ is locally strongly convex and differentiable  whose gradient is $L(h)$-Lipschitz continuous.
The dual problem of (\ref{ssnp}) can be rewritten equivalently as
\begin{equation}\label{ssnd}
\min_{y,z}\big\{h^{\star}(y)+p^{\star}(z)\ | \ \A^{\ast}y+z=c\big\},
\end{equation}
where $h^{\star}$ and $p^{\star}$ are the Fenchel conjugate of $h$ and $p$, respectively.
The assumptions on $h$ imply that $h^{\star}$ is strongly convex  \cite[Proposition 12.60]{RW}, essentially smooth. And its gradient $\nabla h^{\star}$ is locally Lipschitz continuous on $\text{int}(\dom h^{\star})$ with modulus $1/L(h)$ \cite[Corollary 4.4]{GR}. Solving problem (\ref{ssnp}) and its dual (\ref{ssnd}) is equivalent to finding $(\bar{y},\bar{z},\bar{x})$ such that the following Karush-Kuhn-Tucker (KKT) system holds
$$
\mathbf{0}=\nabla h^{\star}(\bar{y})-\A \bar{x}, \quad \mathbf{0}\in-\bar{x}+\partial p^{\star}(\bar{z}), \quad \mbox{and} \quad \mathbf{0}=\A^{\star}\bar{y}+\bar{z}-c.
$$

Given $\sigma_k>0$, the augmented Lagrangian function associated with (\ref{ssnd}) is given by
$$
\L_{\sigma}(y,z;x)=h^{\star}(y)+p^{\star}(z)-\langle x,\A^{*}y+z-c\rangle+\frac{\sigma_k}{2}\big\|\A^{*}y+z-c\big\|_2^{2},
$$
where $x\in\X$ is a multiplier.
Starting from $(y^{(0)},z^{(0)},x^{(0)})\in \text{int}(\dom h^{\star})\times\dom p^{\star}\times\X$, the SSNAL method of \cite{LST} for solving (\ref{ssnd}) takes the following framework
\begin{equation}\label{al}
\left\{
\begin{array}{l}
(y^{(k+1)},z^{(k+1)})\approx \argmin\limits_{y,z}\big\{\Psi_{k}(y,z):=\L_{\sigma_k}(y,z;x^{(k)})\big\},
\\[2mm]
x^{(k+1)}=x^{(k)}-\sigma_k(\mathcal{A}^{\ast}y^{(k+1)}+z^{(k+1)}-c), \text{and} \ \sigma_{k+1}\uparrow \sigma_{\infty}\leq \infty.
\end{array}
\right.
\end{equation}

Since the $(y,z)$-problems don't need to be solved exactly, it is appropriate to use the standard stopping criterion of \cite{Rockafellar1976a, Rockafellar1976b}:
\begin{equation}\label{stop-1}
\Psi_k(y^{(k+1)},z^{(k+1)})-\inf\limits_{y,z}\Psi_k(y,z)\leq\varepsilon_{k}^{2}/2\sigma_{k}, \quad \varepsilon_{k}\geq 0, \quad \sum_{k=0}^{\infty}\varepsilon_{k}<\infty.
\end{equation}

For the global convergence of (\ref{al}) with stopping criterion of (\ref{stop-1}), one can refer to \cite[Theorem 3.2]{LST}. To easily follow this part, we present the convergence results without proofs here.
\begin{theorem}
Suppose that the solution set to (\ref{ssnp}) is nonempty. Let $\{(y^{(k)},\\z^{(k)},x^{(k)})\}$ be an infinite sequence generated by the iterative framework (\ref{al}) with stopping criterion (\ref{stop-1}). Then, the sequence $\{x^{(k)}\}$ is bounded and converges to an optimal solution of (\ref{ssnp}). In addition, $\{(y^{(k)},z^{(k)})\}$ is also bounded and converges to the unique optimal solution $(\bar{y},\bar{z})\in\text{int}(\dom h^{\star})\times\dom p^\star$ of (\ref{ssnd}).
\end{theorem}

The main computational burden of SSNAL lies in solving the augmented Lagrangian subproblem, which is regarded as solving the following problem with fixed $\sigma > 0$ and $\widetilde{x} \in \mathcal{X}$:
$$
\min_{y,z} \ \big\{\Psi(y,z):=\L_{\sigma}(y,z;\widetilde{x})\big\}.
$$
For any $y\in\mathcal{Y}$, we define $\psi(y):=\inf\limits_{z}\Psi(y,z)$. After some simple manipulations, we have
\begin{align*}
\psi(y)=&h^{\star}(y)+p^{\star}\big(\P_{p^{\star}}^{1/\sigma}(\widetilde{x}/\sigma-\mathcal{A}^{\ast}y+c)\big)\\
&+\frac{1}{2\sigma}\big\|\P_{p}^{\sigma}(\widetilde{x}-\sigma(\mathcal{A}^{\ast}y-c))\big\|_2^{2}-\frac{1}{2\sigma}\|\widetilde{x}\|_2^{2}.
\end{align*}
If $(\bar{y},\bar{z})=\argmin\limits_{y,z}\Psi(y,z)$, we can get
$$
\bar{y}=\argmin_y\psi(y) \quad\text{and}\quad \bar{z}=\P_{p^{\star}}^{1/\sigma}(\widetilde{x}/\sigma-\mathcal{A}^{\ast}\bar{y}+c).
$$
Note that $\psi(y)$ is strongly convex and continuously differentiable on $\text{int}(\dom \\h^{\star})$,
thus $\bar{y}$ can be obtained via solving the nonsmooth system
\begin{equation}\label{grad}
\nabla\psi(y)=\nabla h^\star(y)-\mathcal{A}\P_{p}^{\sigma}(\tilde{x}-\sigma(\mathcal{A}^*y-c))=\mathbf{0}, \quad y\in\text{int}(\dom h^{\star}).
\end{equation}
When the generalized Jacobian of $\nabla \psi(y)$ at $y$ was explicitly constructed by using the strongly semismoothness of $\nabla h^\star(\cdot)$ and $\P_{p}^{\sigma}(\cdot)$, then (\ref{grad}) can be solved effectively by SSN method.
To end this part, we list the convergence result of SSN to (\ref{grad}). For its proof, one can refer to \cite[Theorem 3.5]{LST}.
\begin{theorem}
Assume that $\nabla h^\star(\cdot)$ and $\P_{p}^{\sigma}(\cdot)$ are strongly semismooth on $\text{int}(\dom h^{\star})$ and $\mathcal{X}$, respectively.
Let the sequence $\{y^{(l)}\}$ be generated by SSN algorithm. Then $\{y^{(l)}\}$ converges to the solution $y^\infty\in\text{int}(\dom h^{\star})$ of (\ref{grad}) and
$$
\|y^{(l+1)}-y^\infty\|_2=\O(\|y^{(l)}-y^\infty\|_2^{1+\varrho}),
$$
where $\varrho\in(0,1]$.
\end{theorem}

\section{Model Construction and Optimality Condition}\label{model}

This section is devoted to the first assignment of this paper, that is, constructing a PLS with an adaptive lasso regularization for regression parameter and then giving its dual formulation and optimality condition.

\subsection{PLS with Adaptive Lasso for Parameter Regression}

In this section, we restrict our attention to the task of regressing the coefficient $\beta$ of semiparametric regression model (\ref{sprm}).
For this purpose, we should make some preparations for concealing the unknown nonparametric function $g(\cdot)$.
Assuming that $\beta$ is known, then (\ref{sprm}) is simplified to a purely nonparametric regression model
$$
Y_{i}-X_{i}^{\top}\beta=g(T_{i})+\varepsilon_{i},\quad i=1,2,\cdots,n.
$$

In this part, we are particularly interested in using the weighting function method \cite{LTX2008, WJ2003} to estimate the nonparametric part $g(\cdot)$, that is,
$$
g_{n}(t,\beta):=\sum_{i=j}^{n}W_{ni}(t)(Y_{j}-X_{j}^{\top}\beta),
$$
where $W_{ni}(t)=K_{h}(T_{i}-t)\big/\sum_{j=1}^nK_{h}(T_{j}-t)$ is a nonnegative weighting function satisfying $0\leq W_{nj}(t)\leq 1$ and $\sum_{j=1}^nW_{nj}(t)=1$, where $K_{h}(\cdot)=K(\cdot/h)$, $K(\cdot)$ is a nonnegative kernel function and $h$ is a so-called bandwidth which is a constant sequence converging to zero.
Denote
$$
\widetilde{X}_{i}:= X_{i}-\sum_{j=1}^{n}W_{nj}(T_{i})X_{j}\quad\text{and}\quad\widetilde{Y}_{i}:= Y_{i}-\sum_{j=1}^{n}W_{nj}(T_{i})Y_{j},
$$
and replace $g(T_{i})$ in (\ref{sprm}) with $g_{n}(T_{i},\beta)$, we can get a purely parametric regression model as follows:
\begin{equation}\label{linear}
\widetilde{Y}_{i}=\widetilde{X}_{i}^{\top}\beta+\widetilde{\varepsilon}_{i}, \quad i=1,2,\cdots,n,
\end{equation}
where $\widetilde{\varepsilon}_{i}:=\varepsilon_{i}-\sum_{j=1}^{n}W_{nj}(T_{i})\varepsilon_{j}$ is the the residual estimation.
Consider the random sample for $i=1,2,\cdots,n$ as a whole, we can reformulate (\ref{linear}) as the following compact form
\begin{equation}\label{n-linear}
\widetilde{Y}=\widetilde{X}^{\top}\beta+\widetilde{\varepsilon},
\end{equation}
where $\widetilde{Y}\in\mathbb{R}^{n}$, $\widetilde{X}\in\mathbb{R}^{p\times n}$, and $\widetilde{\varepsilon}\in\mathbb{R}^{n}$.

Combining the cross-section least square method of \cite{SPECKMAN} and the adaptive lasso variable selection method for linear regression model of \cite{Zou}, we propose the PLS estimation method for regression parameter $\beta$ as follows:
\begin{equation}\label{pls}
\min_{\beta\in\mathbb{R}^p} \ \big\{\frac{1}{2}\|\widetilde{Y}-\widetilde{X}^{\top}\beta\|_2^{2}+\lambda\sum_{j=1}^{p}\omega_{j}|\beta_{j}|\big\},
\end{equation}
where $\lambda>0$ is a positive parameter and $\omega_{j}>0$ is an adaptive tuning parameter.
For convenience, we assume that the matrix $\widetilde{X}$ is normalized such that the spectral radius of $\widetilde{X}\widetilde{X}^{\top}$ is not greater than $1$, i.e. $\rho(\widetilde{X}\widetilde{X}^{\top})\leq 1$. In this case, the function $\|\widetilde{Y}-\widetilde{X}^{\top}\beta\|_2^{2}$ is convex differentiable whose gradient is $1$-Lipschitz continuous.
It should be noted that the oracle property of (\ref{pls}) can be easily attained by mimicking the proof of \cite{Zou} in which
a PLS method for a purely linear regression with adaptive lasso penalty was considered.
Let $\widehat{\beta}_{pls}:=(\widetilde{X}\widetilde{X}^\top)^{-1}\widetilde{X}\widetilde{Y}$ be the least square estimation. It is from \citep{SPECKMAN} that, $\widehat{\beta}_{pls}$ is a $\sqrt{n}$-consistent estimation to the true parameter $\beta$.
Let  $\widehat{\beta}_{ads}$ be the optimal solution of (\ref{pls}).
Then, based on the common assumptions of semiparametric regression model as listed in \cite{Pls,SPECKMAN}, we can get the oracle property of (\ref{pls}) through a series of derivation. To facilitate readers' understanding, we present the specific assumptions and unproven properties here. For more details on its proof, one may refer to \cite[Theorem 3.3]{Pls}.
\begin{assumption}\label{assump1}
	(a) Decompose $X_{n\times p}:=[X_1,\ldots,X_n]^\top$ into $X=f+\eta$, where $\eta=(\eta_{ij})_{n\times p}$ with $E(\eta_{ij}|T_i)=0$ and $\eta_{ij}$ is independent of $\varepsilon_i$, and $f=(f_{ij})_{n\times p}$ with $f_{ij}:=f_j(T_i)$ satisfying $E(X_{ij}|T_i)=f_j(T_i)$ with unknown smooth function $f_j$ ($j=1,\ldots,p$).
	Suppose that $n^{-1}\eta^{\top}\eta\xrightarrow{P}V$ (convergence in probability) with $V$.\\
   (b) Let $K_{n\times n}=(K)_{ij}$ with $K_{ij}=W_{ni}(T_j)$ and suppose that $tr(K^\top K)=O_p(h^{-1})$, $tr(K)=O_p(h^{-1})$.\\
   (c) Denote $g:=(g(T_1),\ldots,g(T_n))^\top$ and $\widetilde{g}:=(I_n-K)g$, and suppose that $\|\widetilde{g}\|_2^2=O_p(nh^4)$.\\
   (d) Suppose the bandwidth satisfies $h=O(n^{-1/5})$.
\end{assumption}

\begin{theorem}\cite[Theorem 3.3]{Pls}
Let $\beta_{\A}$ be the non-zero coefficient of the true parameter $\beta$ in (\ref{sprm}). Let
$\widehat{\beta}_{ads\mathcal{A}}$ be non-zero coefficient of $\widehat{\beta}_{ads}$.
Let $\mathcal{A}$ and $\mathcal{A}_n$ be the non-zero element index set of real value $\beta$ and estimated value $\widehat{\beta}_{ads}$, respectively. Moreover, assume that the number of non-zeros variates in $\beta$ is $p_0$, i.e., $|\mathcal{A}|=p_0$.
Suppose that $V$ is nonsingular and the tuning parameter is chosen as $\omega_{j}=|\widehat{\beta}_{pls}|^{-\gamma}$ with a $\gamma>0$. If $\lambda/\sqrt{n}\rightarrow 0$ and $\lambda n^{(\gamma-1)/2}\rightarrow\infty$. Then, under the Assumption \ref{assump1}, it holds that\\
(i) $\lim\limits_{n\rightarrow\infty}Pr(\mathcal{A}_{n}=\mathcal{A})=1$,\\
(ii) $\sqrt{n}(\widehat{\beta}_{ads\mathcal{A}}-\beta_\mathcal{A})\xrightarrow{\L}\mathcal{N}(0,\sigma^{2}V_{11}^{-1})$ (convergence in distribution), where $V_{11}\in\mathbb{R}^{p_0\times p_0}$ is a submatrix of $V$.
\end{theorem}

\subsection{Dual formulation and Optimality Condition}

In this part, we analyze the theoretical properties of (\ref{pls}) from the perspective of optimization for subsequent algorithm's developments.
In order to facilitate our analysis, we introduce a pair of auxiliary variables $s:=\widetilde{Y}-\widetilde{X}^{\top}\beta$ and $z:=\beta$.
Then, the problem (\ref{pls}) is reformulated as
\begin{equation}\label{p}
\begin{array}{ll}
\min\limits_{s,z,\beta}& \frac{1}{2}\|s\|_2^{2}+\lambda\sum\limits_{j=1}^{p}\omega_{j}|z_{j}|\\[2mm]
\text{s.t.} &s=\widetilde{Y}-\widetilde{X}^{\top}\beta, \  z=\beta.
\end{array}
\end{equation}

The Lagrangian function associated with (\ref{p}) is defined by
$$
\L(s,z,\beta;u,v):=
\frac{1}{2}\|s\|_2^{2}+\lambda\sum\limits_{j=1}^{p}\omega_{j}|z_{j}|-\langle\widetilde{Y}-\widetilde{X}^{\top}\beta-s,u\rangle
-\langle z-\beta,v\rangle,
$$
where $u\in\mathbb{R}^{n}$ and $v\in\mathbb{R}^{p}$ are multipliers associated with the constraints in (\ref{p}). 
The Lagrangian dual function of problem (\ref{p}) is defined by the minimum value of the Lagrangian function over $(s,z,\beta)$, that is
\begin{align*}
D(u,v)=&\min_{s,z,\beta}\L(s,z,\beta;u,v)\\
=&\min_{s}\big\{\frac{1}{2}\|s\|_2^{2}+\langle s,u\rangle\big\}
+\min_{\beta}\big\{\langle \widetilde{X}^{\top}\beta,u\rangle+\langle\beta,v\rangle\big\}\\
&+\min_{z}\big\{\lambda\sum\limits_{j=1}^{p}\omega_{j}|z_{j}|-\langle z,v\rangle\big\}-\langle \widetilde{Y},u\rangle\\
=&\big\{-\frac{1}{2}\|u\|_2^{2}-\langle \widetilde{Y},u\rangle-\delta_{\B_{\infty}^{(\lambda\omega)}}(v)\ \big| \ \widetilde{X}u+v=\mathbf{0}_{p}\big\},
\end{align*}
where $\B_{\infty}^{(\lambda\omega)}\in \mathbb{R}^p$ is defined by $\B_{\infty}^{(\lambda\omega)}:=\B_{\infty}^{(\lambda\omega_{1})}\times\B_{\infty}^{(\lambda\omega_{2})}\times\cdots\times\B_{\infty}^{(\lambda\omega_{p})}$,
that is to say, $\B_{\infty}^{(\lambda\omega)}=\{v \ | \ |v_{j}|\leq\lambda\omega_{j}, j=1,2,\cdots,p\}$. The indicator function $\delta_{\B_{\infty}^{(\lambda\omega)}}(v)=\mathbf{0}_{p}$ means that $\delta_{\B_{\infty}^{(\lambda\omega_{j})}}(v_{j})=0$ for every $j=1, 2, \cdots, p$.

The Lagrangian dual problem of the original (\ref{p}) is defined by maximizing $D(\cdot)$ over $(u,v)$, which takes the following equivalent from
\begin{equation}\label{d}
\begin{array}{ll}
\min\limits_{u,v}  & \frac{1}{2}\|u\|_2^{2}+\langle \widetilde{Y},u\rangle+\delta_{\B_{\infty}^{(\lambda\omega)}}(v)
\\[2mm]
\text{s.t.} &\widetilde{X}u+v=\mathbf{0}_{p}.
\end{array}
\end{equation}
We note that the aforementioned assumptions on $\widetilde{X}$ implies that $\|u\|_2^{2}$ is strongly convex with modulus $1$ \cite[Proposition 12.60]{RW}.
We say that $(\bar{u},\bar{v})$ is an optimal solution of problem (\ref{d}) if there exists a combination of $(\bar{s},\bar{z},\bar{\beta})$ that is a solution of (\ref{p}) such that
the following KKT system is satisfied
\begin{equation}\label{KKT}
\left\{
\begin{array}{l}
\bar{s}+\bar{u}=\mathbf{0}_{n}, \quad \bar{z}-\bar{\beta}=\mathbf{0}_{p}, \\
\widetilde{Y}-\widetilde{X}^{\top}\bar{\beta}-\bar{s}=\mathbf{0}_{n},\\
\widetilde{X}\bar{u}+\bar{v}=\mathbf{0}_{p},\\
\bar{v} \in \lambda\partial\sum\limits_{j=1}^{p}\omega_{j}|\bar{z}_{j}|.  
\end{array}
\right.
\end{equation}
From \cite[Theorem 23.5]{RR}, we know that the KKT syetem (\ref{KKT}) can be equivalently rewritten as
\begin{equation}\label{KKT2}
\left\{
\begin{array}{l}
\widetilde{Y}-\widetilde{X}^{\top}\bar{\beta}+\bar{u}=\mathbf{0}_{n},\quad
\widetilde{X}\bar{u}+\bar{v}=\mathbf{0}_{p},\\
\bar{v}_{j}=\Pi_{\B_{\infty}^{(\lambda\omega_{j})}}(\bar{v}_{j}+\bar{\beta}_{j}), \quad  j=1,2,\cdots,p,
\end{array}
\right.
\end{equation}
where
$\bar{v}_{j}=\Pi_{\B_{\infty}^{(\lambda\omega_{j})}}(\bar{v}_{j}+\bar{\beta}_{j})
=\min\big\{\lambda\omega_{j}, \max\{\bar{v}_{j}+\bar{\beta}_{j}, -\lambda\omega_{j}\}\big\}$.

\section{SSNAL Method for Dual Problem (\ref{d})}\label{algorithm}

In this section, we consider selecting the regression parameter $\beta$ via the PLS (\ref{pls}) as well as its dual (\ref{d}).
We employ SSNAL method on (\ref{d}) where SSN is used to solve the involved semismooth equations.
\subsection{ Algorithm's Construction and Convergence Theorem}

Given $\sigma_k>0$, the augmented Lagrangian function associated with (\ref{d}) is defined by
\begin{align}\label{alf}
\L_{\sigma}(u,v;\beta)
=&\frac{1}{2}\|u\|_2^{2}+\langle\widetilde{Y}, u\rangle+\delta_{\B_{\infty}^{(\lambda\omega)}}(v)\notag\\
&-\langle\widetilde{X}u+v, \beta\rangle
+\frac{\sigma_k}{2}\|\widetilde{X}u+v\|_2^{2},
\end{align}
where $\beta\in\mathbb{R}^{p}$ is a multiplier or the $p$-variate regression coefficient in problem (\ref{pls}). 
While the SSNAL method of (\ref{al}) is employed on the problem (\ref{d}), its detailed steps can be summarized as follows:
{\small
\begin{framed}
\noindent
{\bf Algorithm SSNAL: A inexact augmented Lagrangian method for (\ref{d})}
\vskip 1.0mm \hrule \vskip 1mm
\noindent
\textbf{Step 1.} Take $\sigma_{0}>0$,
$(u^{(0)},v^{(0)},\beta^{(0)})\in\mathbb{R}^{n}\times\B_{\infty}^{(\lambda\omega)}\times\mathbb{R}^{p}$.
For $k=0,1,\ldots$, do the following operations iteratively.\\
\textbf{Step 2.} Compute
\begin{align}\label{subproblem}
&(u^{(k+1)},v^{(k+1)})\approx
\argmin_{u,v}\big\{\Psi_{k}(u,v):=\L_{\sigma_{k}}(u,v;\beta^{(k)})\big\}\\
&=\argmin_{u,v}\big\{\frac{1}{2}\|u\|_2^{2}+\langle\widetilde{Y}, u\rangle+\delta_{\B_{\infty}^{(\lambda\omega)}}(v)
-\langle\widetilde{X}u+v, \beta^{(k)}\rangle+\frac{\sigma_{k}}{2}\|\widetilde{X}u+v\|_2^{2}\big\}.\nonumber
\end{align}
\textbf{Step 3.} Compute $\beta^{(k+1)}=\beta^{(k)}-\sigma_{k}\big(\widetilde{X}u^{(k+1)}+v^{(k+1)}\big)$ and update $\sigma_{k+1}\uparrow\sigma_{\infty}\leq\infty$.
\end{framed}
}

Since the inner problem (\ref{subproblem}) is not expected to be solved exactly, we may use the standard stopping criterion studied in \citep{Rockafellar1976a,Rockafellar1976b} to derive an inexact solution, that is
\begin{equation}\label{stop}
\Psi_{k}(u^{(k+1)},v^{(k+1)})-\inf_{u,v}\Psi_{k}(u,v)\leq\pi_{k}^{2}/2\sigma_{k}, \quad  \pi_{k}\geq 0, \quad \sum\limits_{k=0}^{\infty}\pi_{k}<\infty.
\end{equation}
From \cite[Theorem 3.2]{LST}, the global convergence of SSNAL with a sketched proof can be described as follows:
\begin{theorem}
Suppose that the solution set to (\ref{pls}) is nonempty. Let $\{(u^{(k)},\\v^{(k)},\beta^{(k)})\}$ be the infinite sequence generated by SSNAL method with stopping criterion (\ref{stop}). Then, the sequence $\{\beta^{(k)}\}$ is bounded and converges to an optimal solution of (\ref{pls}). In addition, the sequence $\{(u^{(k)},v^{(k)})\}$ is also bounded and converges to the unique optimal solution $(\bar u, \bar v)\in\mathbb{R}^{n}\times\B_{\infty}^{(\lambda\omega)}$ of (\ref{d}).
\end{theorem}
\begin{proof}
The nonempty assumption on the solution set of (\ref{pls}) indicates that the optimal value of (\ref{pls}) is finite.
Besides, by Fenchel's duality theorem \cite[Corollary 31.2.1]{RR}, the solution set to (\ref{d}) is nonempty and the optimal value of (\ref{d}) is finite and equal to the optimal value of (\ref{pls}).
That is to say, the solution set to KKT system (\ref{KKT2}) is nonempty.
By the strong convexity of $\|u\|^{2}$, the uniqueness of the optimal solution $(\bar{u},\bar{v})\in\mathbb{R}^{n}\times\B_{\infty}^{(\lambda\omega)}$ of (\ref{d}) can obtain directly.
Combining this uniqueness with \cite[Theorem 4]{Rockafellar1976a}, we can easily obtain the boundedness of $\{(u^{(k)},v^{(k)})\}$ and other desired results readily.
\end{proof}
\subsection{Solving the Augmented Lagrangian Subproblems}

This part is devoted to employing the SSN to solve the inner subproblems (\ref{subproblem}) resulting from the augmented Lagrangian method.
With fixed $\sigma>0$ and $\widetilde{\beta}\in\mathbb{R}^{p}$, it aims to solving
\begin{equation}\label{newton}
\min_{u,v}\Psi(u,v):=\L_{\sigma}(u,v;\widetilde{\beta}).
\end{equation}
It is evident that $\Psi(u,v)$ is a strongly convex function.

Combining the strong convexity of $\Psi(\cdot,\cdot)$, we have that for any $\alpha\in\mathbb{R}$, the level set $\Psi_{\alpha}:=\{(u,v)\in\mathbb{R}^{n}\times\B_{\infty}^{(\lambda\omega)}\ | \ \Psi(u,v)\leq\alpha\}$ is closed, convex and bounded, which means that (\ref{newton}) admits an unique optimal solution $(\bar{u},\bar{v})$.
For any $u\in\mathbb{R}^{n}$, denote $\psi(u):=\inf_{v}\Psi(u,v)$. Then, we have
\begin{align*}
\psi(u)=&\inf_{v}\Big\{\frac{1}{2}\|u\|_2^{2}+\langle\widetilde{Y}, u\rangle
-\frac{1}{2\sigma}\|\widetilde{\beta}\|_2^{2}+\delta_{\B_{\infty}^{(\lambda\omega)}}(v)
+\frac{\sigma}{2}\big\|v-(\widetilde{\beta}/\sigma-\widetilde{X}u)\big\|_2^{2}\Big\}\\[2mm]
=&\frac{1}{2}\|u\|_2^{2}+\langle\widetilde{Y}, u\rangle-\frac{1}{2\sigma}\|\widetilde{\beta}\|_2^{2}
+\inf_{v\in\B_{\infty}^{(\lambda\omega)}}\Big\{\frac{\sigma}{2}\big\|v-(\widetilde{\beta}/\sigma-
\widetilde{X}u)\big\|_2^{2}\Big\}\\[2mm]
=&\frac{1}{2}\|u\|_2^{2}+\langle\widetilde{Y}, u\rangle-\frac{1}{2\sigma}\|\widetilde{\beta}\|_2^{2}
+\frac{\sigma}{2}\big\|\P_{\|\cdot\|_{1}}^{\lambda\omega}(\widetilde{\beta}/\sigma-\widetilde{X}u)\big\|_2^{2},
\end{align*}
where the last equality is from $x-\Pi_{\B_{\infty}^{(\lambda\omega)}}(x)=x-\P_{\delta_{\B_{\infty}^{(\lambda\omega)}}}(x)=\P_{\|\cdot\|_{1}}^{\lambda\omega}(x)$,
$\big(\P_{\|\cdot\|_{1}}^{\lambda\omega}(\widetilde{\beta}/\sigma-\widetilde{X}u)\big)_{j}=
\P_{|\cdot|}^{\lambda\omega_{j}}(\widetilde{\beta}_{j}/\sigma-(\widetilde{X}u)_{j})
=\text{sgn}(\widetilde{\beta}_{j}/\sigma-(\widetilde{X}u)_{j})\odot\max\big\{|\widetilde{\beta}_{j}/\sigma-(\widetilde{X}u)_{j}|
-\lambda\omega_{j}, 0\big\}$ for any $j=1,2,\cdots,p$.
Therefore, if $(\bar{u},\bar{v})=\argmin\Psi(u,v)$, then we can get that
$$
\bar{u}=\argmin\psi(u) \quad\text{and}\quad \bar{v}=\Pi_{\B_{\infty}^{(\lambda\omega)}}(\widetilde{\beta}/\sigma-\widetilde{X}\bar{u}),
$$
where
$\bar{v}_{j}=\Pi_{\B_{\infty}^{(\lambda\omega_{j})}}(\widetilde{\beta}_{j}/\sigma-(\widetilde{X}\bar{u})_{j})
=\min\big\{\lambda\omega_{j}, \max\{\widetilde{\beta}_{j}/\sigma-(\widetilde{X}\bar{u})_{j}, -\lambda\omega_{j}\}\big\}$
for any $j=1, 2, \cdots, p$.
Note that $\psi(\cdot)$ is strongly convex and continuously differentiable with gradient
$$
\nabla\psi(u)=u+\widetilde{Y}-\widetilde{X}^{\top}\P_{\|\cdot\|_{1}}^{\sigma\lambda\omega}(\widetilde{\beta}-\sigma\widetilde{X}u),
$$
then $\bar{u}$ can be obtained by solving the nonsmooth equation
\begin{equation}\label{gra}
\nabla\psi(u)=\mathbf{0}.
\end{equation}

Let $u\in\mathbb{R}^{n}$ be any given point, define
$$
\widehat{\partial}^{2}\psi(u)=I_{n}
+\sigma \widetilde{X}^{\top}\partial\P_{\|\cdot\|_{1}}^{\sigma\lambda\omega}(\widetilde{\beta}-\sigma\widetilde{X}u)\widetilde{X},
$$
where $\partial\P_{\|\cdot\|_{1}}^{\sigma\lambda\omega}(\widetilde{\beta}-\sigma\widetilde{X}u)$ is the Clarke subdifferential of the Lipschitz continuous mapping $\P_{\|\cdot\|_{1}}^{\sigma\lambda\omega}(\cdot)$ at point $\widetilde{\beta}-\sigma\widetilde{X}u$.
From \cite[Proposition 2.3.3 and Theorem 2.6.6]{FH}, we know that
$$
\partial^{2}\psi(u)d\subseteq\widehat{\partial}^{2}\psi(u)d, \quad \forall d\in \mathbb{R}^{n},
$$
where $\partial^{2}\psi(u)$ is the generalized Hessian of $\psi$ at $u$.
Define
\begin{equation}
H:=I_{n}+\sigma\widetilde{X}^{\top}\Theta\widetilde{X},
\end{equation}
with $\Theta\in\partial\P_{\|\cdot\|_{1}}^{\sigma\lambda\omega}(\widetilde{\beta}-\sigma\widetilde{X}u)$.
Then, we have $H\in\widehat{\partial}^{2}\psi(u)$.
Note that $I_{n}$ is a $n$-dimensional identity matrix and $\Theta$ is a sparse $0$-$1$ diagonal matrix. It then follows that $H$ is symmetric and positive definite.

It is widely known that continuous piecewise affine functions and twice continuously differentiable functions are all strongly semismooth everywhere. Therefore, $\P_{\|\cdot\|_{1}}^{\sigma\lambda\omega}(\cdot)$ is strongly semismooth. Thus, we can employ the SSN algorithm to solve the semismooth nonlinear equations (\ref{gra}).
The convergence results for SSN algorithm are stated in the following theorem.
\begin{theorem}
Let the sequence $\{u^{(l)}\}$ be generated by SSN algorithm. Then $u^{(l)}$ converge to the unique optimal solution $u^\infty\in\mathbb{R}^{n}$ of the problem in (\ref{gra}) and
$$
\|u^{(l+1)}-u^\infty\|_2=\O\big(\|u^{(l)}-u^\infty\|_2^{1+\varrho}\big),
$$
where $\varrho\in(0,1]$.
\end{theorem}

We now discuss the implementation of stopping criterion (\ref{stop}) for SSN algorithm to solve the subproblem (\ref{subproblem}) in SSNAL.
In fact, we notice that
$$
\Psi_{k}(u^{(k+1)},v^{(k+1)})=\inf_{v}\Psi_{k}(u^{(k+1)},v)=\psi_{k}(u^{(k+1)}),
$$
$$
\inf\Psi_{k}=\inf_{u,v}\Psi_{k}(u,v)=\inf_{u}\inf_{v}\Psi_{k}(u,v)=\inf_{u}\psi_{k}(u)=\inf\psi_{k},
$$
which implies that $\Psi_{k}(u^{(k+1)},v^{(k+1)})-\inf\Psi_{k}=\psi_{k}(u^{(k+1)})
-\inf\psi_{k}$.
Let $\widehat{u}=\argmin\psi_{k}(u)$, by the strong convexity of $\psi_{k}$, we have
$$
\psi_{k}(\widehat{u})-\psi_{k}(u^{(k+1)})\geq\langle \nabla\psi_{k}(u^{(k+1)}), \widehat{u}-u^{(k+1)}\rangle
+\frac{1}{2}\|\widehat{u}-u^{(k+1)}\|_2^{2},
$$
then
\begin{align*}
\psi_{k}(u^{(k+1)})-\psi_{k}(\widehat{u})\leq
&-\Big(\langle\nabla\psi_{k}(u^{(k+1)}), \widehat{u}-u^{(k+1)}\rangle+\frac{1}{2}\|\widehat{u}-u^{(k+1)}\|_2^{2}\Big)\\[2mm]
=&-\frac{1}{2}\big\|\widehat{u}-u^{(k+1)}+\nabla\psi_{k}(u^{(k+1)})\big\|_2^{2}
+\frac{1}{2}\big\|\nabla\psi_{k}(u^{(k+1)})\big\|_2^{2}\\[2mm]
\leq &\frac{1}{2}\big\|\nabla\psi_{k}(u^{(k+1)})\big\|_2^{2}.
\end{align*}
Therefore, we know
$$
\Psi_{k}(u^{(k+1)},v^{(k+1)})-\inf\Psi_{k}=\psi_{k}(u^{(k+1)})
-\inf\psi_{k}\leq \frac{1}{2}\big\|\nabla\psi_{k}(u^{(k+1)})\big\|_2^{2}.
$$
The stopping criterion (\ref{stop}) can be achieved by the following implementable criterion
\begin{equation}
\big\|\nabla\psi_{k}(u^{(k+1)})\big\|_2\leq\sqrt{1/\sigma_{k}}\pi_{k}, \quad \pi_{k}\geq0, \quad
\sum_{k=0}^{\infty}\pi_{k}<\infty.
\end{equation}
That is, the stopping criterion (\ref{stop}) will be satisfied as long as $\|\nabla\psi_{k}(u^{(k+1)})\|_2$ is sufficiently small.

In summary, the iterative framework of SSNAL method for dual problem (\ref{d}) can be listed as follows:
{\small
\begin{framed}
\noindent
{\bf Algorithm: SSNAL}
\vskip 1.0mm \hrule \vskip 1mm
\noindent
\begin{itemize}
\item[Step 0.] Given $\sigma_{0}>0$, $\mu\in(0,1/2)$, $\bar{\eta}\in(0,1)$, $t\in(0,1]$, and $\rho\in(0,1)$. Choose
$(u^{(0)},v^{(0)},\beta^{(0)})\in\mathbb{R}^{n}\times\B_{\infty}^{(\lambda\omega)}\times\mathbb{R}^{p}$.
For $k=0,1,\ldots$, do the following operations iteratively.
\item[Step 1.] Choose $\tilde{u}^{(0)}:=u^{(k)}$. Do the following operations iteratively.
\begin{itemize}
\item[Step 1.1.] Choose $\Theta_{l}\in\partial\P_{\|\cdot\|_{1}}^{\sigma\lambda\omega}(\beta^{(k)}-\sigma_k\widetilde{X}\tilde{u}^{(l)})$.
Let $H_{l}:=I_{n}+\sigma_{k}\widetilde{X}^{\top}\Theta_{l}\widetilde{X}$.
\item[Step 1.1.] Solve the linear system
\begin{equation}\label{sublin}
H_{l}d+\nabla\psi(\tilde{u}^{(l)})=\mathbf{0}
\end{equation}
exactly or by the conjugate gradient (CG) algorithm to find $d^{i}$ such that
$$
\big\|H_{l}d^{l}+\nabla\psi(\tilde{u}^{(l)})\big\|_2\leq\min(\bar{\eta},\|\nabla\psi(\tilde{u}^{(l)})\|_2^{1+t}).
$$
\item[Step 1.2.] (Line search) Set $\alpha_{l}=\rho^{m_{l}}$, where $m_{l}$ is the first nonnegative integer $m$ such that
$$
\tilde{u}^{(l)}+\rho^{m}d^{l}\in\mathbb{R}^{n} \quad \text{and} \quad \psi(\tilde{u}^{(l)}+\rho^{m}d^{l})\leq\psi(\tilde{u}^{(l)})+\mu\rho^{m}\langle\nabla\psi(\tilde{u}^{(l)}),d^{l}\rangle.
$$
\item[Step 1.3.] Compute
$$
\tilde{u}^{(l+1)}=\tilde{u}^{(l)}+\alpha_{i}d^{l}.
$$
\end{itemize}
\item[Step 2.] Let $u^{(k+1)}:=\tilde{u}^{(l+1)}$ and compute $v^{(k+1)}$ component-wise via
$$
v^{(k+1)}_{j}=\min\Big\{\lambda\omega_{j}, \max\big\{\beta^{(k)}_{j}/\sigma_{k}-(\widetilde{X}u^{(k+1)})_{j}, -\lambda\omega_{j}\big\}\Big\},
\quad j=1,2,\cdots,p.
$$
\item[Step 3.] Compute
$$
\beta^{(k+1)}=\beta^{(k)}-\sigma_{k}\big(\widetilde{X}u^{(k+1)}+v^{(k+1)}\big),
$$
and update $\sigma_{k+1}\uparrow\sigma_{\infty}<\infty$.
\end{itemize}
\end{framed}
}

At the end of this section, similarly to \cite{LST}, we significantly reduce the computational cost by performing an in-depth analysis of the coefficient matrix in (\ref{sublin}). Specifically, (\ref{sublin}) essentially has the following form:
\begin{equation}\label{New-lin}
(I_{n}+\sigma\widetilde{X}^{\top}\Theta\widetilde{X})d=-\nabla\psi(u),
\end{equation}
where the cost of computing $\widetilde{X}^{\top}\Theta\widetilde{X}$ is $\O(n^2p)$.
Denote $\Theta:=\diag(\theta)$, and the $i$-th diagonal element $\theta_i$ is given by
$$
\theta_i=
\left\{
\begin{array}{l}
1, \quad  |\theta_{i}|>\sigma\lambda\omega_i, \\
0, \quad  |\theta_{i}|\leq\sigma\lambda\omega_i,
\end{array}
\right.
\quad  i=1,2,\cdots,p.
$$
Obviously, $\Theta$ is a special diagonal matrix with elements $0$ or $1$ on its diagonal position.
Let $\D$ be the index set such that $\Theta_{ii}=1$, i.e., $\D:=\{i\ | \ |\theta_i|>\sigma\lambda\omega_i, i=1,2,\cdots,p\}$, and the cardinality of $\D$ is denoted by $r$, i.e., $r=|\D|$.
Let $\widetilde{X}_{\D}\in\mathbb{R}^{r\times n}$ be the submatrix of $\widetilde{X}$ with rows in $\D$.
Then, we have
$$
\widetilde{X}^{\top}\Theta\widetilde{X}=(\widetilde{X}^\top\Theta)(\widetilde{X}^\top\Theta)^\top=\widetilde{X}^\top_{\D}\widetilde{X}_{\D},
$$
which means the cost of computing $\widetilde{X}^{\top}\Theta\widetilde{X}$ is reduced to $\O(n^2r)$.
The inverse of $I_{n}+\sigma\widetilde{X}^{\top}\Theta\widetilde{X}$ admits an explicit form \citep{SMW} of
$$
(I_{n}+\sigma\widetilde{X}^{\top}\Theta\widetilde{X})^{-1}=(I_{n}+\sigma\widetilde{X}^\top_{\D}\widetilde{X}_{\D})^{-1}
=I_n-\widetilde{X}_{\D}^{\top}(\sigma^{-1}I_r+\widetilde{X}_{\D}\widetilde{X}_{\D}^{\top})^{-1}\widetilde{X}_{\D},
$$
which is determined by inverting a much smaller $r\times r$ matrix.
In this case, the total computational cost for solving the Newton linear system (\ref{New-lin}) is $\O(r^2(n+r))$, which is greatly reduced because $r$ is sufficiently small.
\section{Numerical Experiments}\label{numer}

In this section, we use random synthetic and real data to highlight the advantages of the semiparametric regression method (\ref{pls}) with adaptive lasso penalty and
highlight the numerical performance of SSNAL method.
Specifically, we consider both low-dimensional ($p<n$) and high-dimensional ($p>n$) cases in the simulation experiments.
In each case, we use an example to illustrate the progressiveness  of the SSNAL method, and then test against the popular ADMM for performance comparison.
We also test SSNAL and ADMM with a real data set to evaluate the algorithms' practical performance.
All the experiments are performed with Microsoft Windows 10 and MATLAB R2019a, and run on a PC with an Intel Core i7-9700 CPU at 3.00 GHz and 16 GB of memory.

\subsection{Brief Description of ADMM  for Problem (\ref{d})}

The ADMM aims to minimize the augmented Lagrangian function (\ref{alf}) with respect to $u$, then to $v$, and update the multiplier $\beta$ immediately. That is,
$$
\left\{
\begin{array}{l}
u^{(k+1)}=\argmin\limits_{u}\L_{\sigma}(u,v^{(k)};\beta^k),\\[3mm]
v^{(k+1)}=\argmin\limits_{v}\L_{\sigma}(u^{k+1},v;\beta^{(k)}),\\[3mm]
\beta^{(k+1)}=\beta^{(k)}-\tau\sigma(\widetilde{X}u^{(k+1)}+v^{(k+1)}),
\end{array}
\right.
$$
where $\tau\in(0,\frac{1+\sqrt{5}}{2})$ is the step size.
It is trivial to deduce that each subproblem admits an explicit solution, which makes the algorithms easily implementable.
The iterative framework of ADMM for problem (\ref{d}) is as follows, with implementation details omitted for the sake of simplicity.
It should be noted that, in the following test, we choose $\tau=1.618$, which has been numerically proven to achieve better performance.
Lastly, for the convergence of ADMM, one may refer to \cite[Theorem B1]{SEMP13}.

\begin{framed}
\noindent
{\bf Algorithm: ADMM}
\vskip 1.0mm \hrule \vskip 1mm
\noindent
\textbf{Step 0.} Given $\sigma>0$ and $\tau\in(0,\frac{1+\sqrt{5}}{2})$. Choose
$(v^{(0)},\beta^{(0)})\in\B_{\infty}^{(\lambda\omega)}\times\mathbb{R}^{p}$.
For $k=0,1,\ldots$, do the following operations iteratively.\\[2mm]
\textbf{Step 1.} Compute
$$
u^{(k+1)}=(I+\sigma\widetilde{X}^{\top}\widetilde{X})^{-1}(\widetilde{X}^{\top}\beta^{(k)}-\widetilde{Y}-\sigma\widetilde{X}^{\top}v^{(k)}).
$$
\textbf{Step 2.} Compute
$$
v^{(k+1)}_{j}=\min\big\{\lambda\omega_{j}, \max\{\beta^{(k)}_{j}/\sigma-(\widetilde{X}u^{(k+1)})_{j}, -\lambda\omega_{j}\}\big\},
\quad j=1,2,\cdots,p.
$$
\textbf{Step 3.} Update
$$
\beta^{(k+1)}=\beta^{(k)}-\tau\sigma(\widetilde{X}u^{(k+1)}+v^{(k+1)}).
$$
\end{framed}

\subsection{Simulation Study}

\subsubsection{Experiments' Setup}

The values of bandwidth $h$, parameter $\lambda$, and weight vector $\omega$ play key roles in the implementation of SSNAL and ADMM.
The bandwidth $h$ is selected by means of cross-validation criterion.  For more details on selecting the bandwidth, one may refer to the book of \cite{FG1996}.
There are many effective methods to select the parameter $\lambda$, e.g., \cite{JJL2015,WLT2007}.
In this test, we follow the continuation technique of \cite{JJL2015} to set an interval $[\lambda_{\min}, \lambda_{\max}]$, where $\lambda_{\max}=\frac{1}{2}\|\widetilde{X}\widetilde{Y}\|^2_{\infty}$ and $\lambda_{\min}=1e-10\lambda_{\max}$.
Then, we employ an equally distributed partition on a logarithmic scale to divide this interval into 200 subintervals. We then use BIC \citep{KK2007} and HBIC \citep{WKL2013} to select a proper regularization parameter $\lambda$ for low-dimensional and high-dimensional cases, respectively.
For the weight vector $\omega$, a smaller weight for larger $|\beta_j|$ corresponds to a smaller bias or even an unbiased estimator, while a larger weight for smaller $|\beta_j|$ leads to a more simplified model.
Inspired by the work of \cite{Zou}, we select $\omega$ by two different approaches.
For low-dimensional case, we denote $\widehat{\beta}_{LS}:=(\widetilde{X}\widetilde{X}^{\top})^{-1}\widetilde{X}\widetilde{Y}$ and then choose $\omega_j=|(\widehat{\beta}_{LS})_j|^{-2}, j=1, 2, \cdots, p$.
For high-dimensional case,  we denote $\mathcal{J}:=\{j | (\beta^*)_j\neq 0\}$ and let $\widetilde{X}_{\mathcal{J}}\in \mathbb{R}^{|\mathcal{J}|\times n}$, and then generate $\widehat{\beta}_{LS}$ by $\widehat{\beta}_{LS}(\mathcal{J})=(\widetilde{X}_{\mathcal{J}}\widetilde{X}_{\mathcal{J}}^{\top})^{-1}\widetilde{X}_{\mathcal{J}}\widetilde{Y}$.
We set the remaining elements in $\widehat{\beta}_{LS}$ are all $1e-3$ and then choose $\omega_j=|(\widehat{\beta}_{LS})_j|^{-2}$ for $j=1, 2, \cdots, p$.
In this experiment, we generate $T_i (i=1, \cdots, n)$ from the uniform distribution on $[0,1]$ and generate the random errors $\varepsilon_i \sim \mathcal{N}(0,1)$. We uniformly use the kernel function $K(x)=\frac{3}{4}(1-x^2)$ with $|x|\leq 1$. For other parameters in SSNAL, we choose $\mu=0.1$, $\rho=0.8$, $\sigma_0=0.01$, and the largest $\sigma_\infty$ is $2$. Other parameters' values will be given when they occurs.

Recalling that the task of the SSN method, as stated in Step 1, is to solve the nonsmooth equations $\nabla\psi(u)=\mathbf{0}$.
In this test, we terminate the inner loop when $\nabla \psi(u^{k+1})<10^{-6}$ to produce an inexact solution.
Additionally, according to the KKT condition in (\ref{KKT2}), the stopping rule for SSNAL and ADMM is set as follows:
\begin{equation}\label{kkts}
\text{Res}:=\frac{\|\beta^{(k)}-\P_{\|\cdot\|_1}^{\lambda \omega}(\beta^{(k)}-\widetilde{X}(\widetilde{X}^{\top}\beta^{(k)}-\widetilde{Y}))\|_2}{1+\|\beta^{(k)}\|_2+\|\widetilde{X}^{\top}\beta^{(k)}-\widetilde{Y}\|_2}<10^{-6},
\end{equation}
where $\text{Res}$ is regarded as the relative KKT residual.
Moreover, the iterative process will be forcefully terminated when the maximum number of iterations ($20$ for SSNAL and $2000$ for ADMM) is reached without achieving convergence. In addition, to evaluate the performance of each algorithm, we mainly use the following metrics:\
\begin{itemize}
  \item \text{NNZ}: the estimated number of non-zero elements which is defined by $\text{NNZ}:=\min\Big\{k| \sum_{i=1}^{k}|\widetilde{\beta}_i|\geq 0.999 \|\widetilde{\beta}\|_1\Big\}$ where $\widetilde{\beta}$ is obtained by sorting the estimated $\bar{\beta}$ such that $|\bar{\beta}_1|\geq \cdots \geq |\bar{\beta}_p|$;
  \item \text{ReErr}: the relative error which is defined by $\text{ReErr}:=\frac{\|{\beta}^*-\bar{\beta}\|_2}{\|\beta^*\|_2}$;
  \item \text{Res}: the KKT residual;
  \item \text{Time(s)}: the running time in second;
  \item \text{Iter}: the number of iterations.
\end{itemize}
At last, we emphasize that all numerical results listed in this section are the averages of $20$ repeated experiments.

\subsubsection{Low-dimensional Case (p$<$n)}

In this part, we generate the matrix $X=(X_1,\cdots,X_n)\in \mathbb{R}^{p\times n}$ by the way that each column of $X$ comes from $\mathcal{N}(0,\Sigma)$, where $\Sigma_{i,j}=0.7^{|i-j|}, 1\leq i,j \leq n$.
The measurable function in model (\ref{sprm}) is selected as $g(x)=\sin(2 \pi x)$ with $x\in [0,1]$.

The first task is to visibly evaluate the effectiveness of SSNAL for low-dimensional regression problems.
For our purpose, we consider the simulation results with $n=10000$ and $p=500$. In this test, we consider the case where the underlying regression coefficient $\beta^*$ in model (\ref{sprm}) only contains $10$ number of non-zero component with fixed position, that is
$\beta^*_i=0$ except for $\beta^*_{55}=9$, $\beta^*_{83}=-5$, $\beta^*_{96}=-7$, $\beta^*_{251}=3$, $\beta^*_{315}=-6$, $\beta^*_{368}=1$, $\beta^*_{404}=10$, $\beta^*_{456}=-8$, $\beta^*_{465}=2$, and $\beta^*_{482}=7$. We show the results estimated by SSNAL in the form of a box plot in Figure~\ref{fig1}, where the boxes reflect the dispersion for estimated regression coefficients from 20 experiments. It can be clearly seen from this plot that, in this low-dimensional case, SSNAL can accurately find the positions of the non-zero elements, and can almost correctly estimate the values of the non-zero elements. Additionally, to further highlight the estimation performance of the SSNAL algorithm, we provide the values of relative error and variance above Figure~\ref{fig1}. The very small values of these metrics indicate that SSNAL can produce highly accurate and stable results in this low-dimensional scenario.

\begin{figure}[h]
\centering
\includegraphics[width=4.5in]{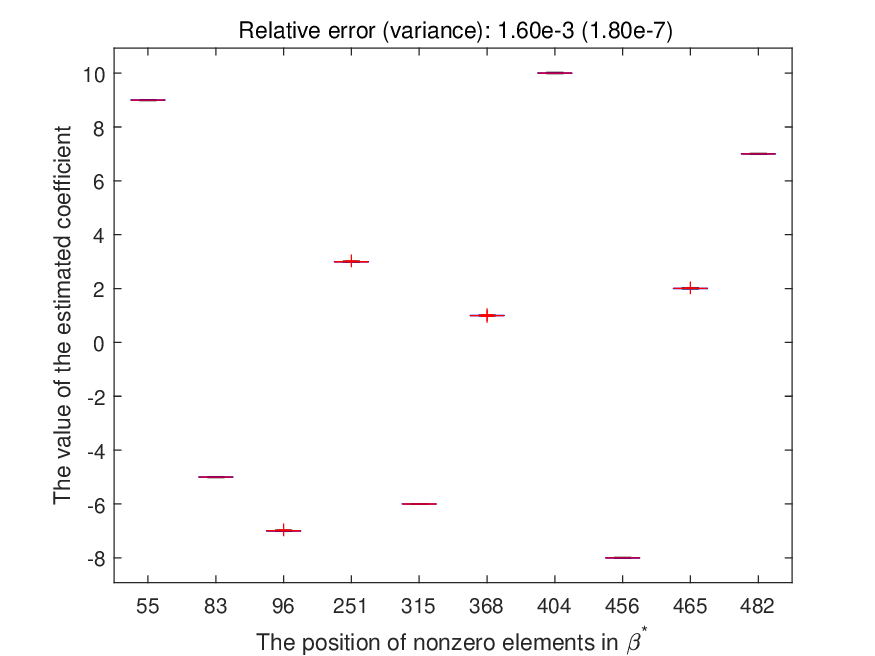}
\caption{{\small The calculation effect of SSNAL with $n=10000$, $p=500$, and $\lambda=0.5$.}}
\label{fig1}
\end{figure}

The second task is to compare the performance of SSNAL and ADMM for (\ref{pls}) with the standard lasso penalty and adaptive lasso penalty, named SSNAL$_a$, SSNAL$_l$, ADMM$_a$ and ADMM$_l$, respectively.
In this test, the true coefficient $\beta^*$ is generated by setting some components to be uniformly distributed within an interval, while others are set to zero.
The fixed interval is $[0,20]$ and the number of non-zero elements in $\beta^*$ is $20$.
For model (\ref{sprm}), the number of samples is set to $n=1000$ and the dimension is set to $p=500$.
We run SSNAL and ADMM $20$ times, and the average results are listed in Table~\ref{tab:1}.

\begin{table}[h]
\centering
\caption{Comparison results of SSNAL and ADMM  with $n=1000$ and $p=500$.}
\label{tab:1}
\setlength{\tabcolsep}{1.5pt}
\begin{tabular}{c|ccccc}
\hline\noalign{\smallskip}
Methods & ReErr & NNZ & Res & Time(s) & Iter\\
\noalign{\smallskip}\hline\noalign{\smallskip}
SSNAL$_a$ & 7.50e-3 (1.40e-3)& 20 (0) & 3.71e-8 (8.95e-9) & 0.29 (0.07) & 4 (0) \\
\specialrule{0em}{3pt}{3pt}
SSNAL$_l$  & 2.37e-2 (3.90e-3)& 21.6 (1.51) & 3.44e-8 (7.79e-9) & 1.11 (0.20) & 5 (0)\\
\specialrule{0em}{3pt}{3pt}
ADMM$_a$ & 7.50e-3 (1.40e-3) & 20 (0) & 9.89e-7 (5.69e-9)   & 73.68 (6.88) & 608.50 (23.38)\\
\specialrule{0em}{3pt}{3pt}
ADMM$_l$  & 2.37e-2 (3.90e-3) & 21.6 (1.51) & 9.85-7 (5.35e-9) & 70.84 (3.26) & 604.40 (16.19)\\
\noalign{\smallskip}\hline
\end{tabular}
\end{table}

From the table, we can see that the values of $\text{ReErr}$ obtained using adaptive lasso are consistently lower than those obtained with standard lasso. The methods using adaptive lasso successfully identify all the non-zero components, whereas standard lasso fails to do so. This phenomenon is consistent with the well-known theoretical results in the literature, which indicate that adaptive lasso has desirable oracle properties. Additionally, both SSNAL and ADMM successfully estimate the regression coefficients within a finite number of iterations. The last two columns show that the computing time and the number of iterations needed by SSNAL is greatly less than those of ADMM, which demonstrates that the PLS method is effective and SSNAL is highly advanced.

\subsubsection{High-dimensional Case (p$>$n)}

In this part, we generate a $p\times n$ random Gaussian matrix $\bar{X}$ whose entries are i.i.d. $\sim\mathcal{N}(0,1)$. Then the design matrix $X$ is generated by setting $X_1=\bar{X}_1$, $X_n=\bar{X}_n$, and $X_j=\bar{X}_j+0.7*(\bar{X}_{j+1}+\bar{X}_{j-1})$ for $j=2, \cdots, n-1$.
Different to the lower-dimension case, the measurable function in model (\ref{sprm}) is chosen as $g(x)=\cos(2 \pi x)$ with $x\in [0,1]$.

The first task in this part is to illustrate the effectiveness of SSNAL in a high-dimensional case, i.e., $n=300$ and $p=10000$.
In this test, we consider the case where the underlying regression coefficient $\beta^*$ in model (\ref{sprm}) only contains $10$ number of non-zero component with fixed position, that is
$\beta^*_i=0$  except for $\beta^*_{104}=5$, $\beta^*_{572}=2$, $\beta^*_{1746}=-4$, $\beta^*_{2947}=-3$, $\beta^*_{4065}=-5$, $\beta^*_{5092}=4$, $\beta^*_{5112}=-1$, $\beta^*_{6680}=1$, $\beta^*_{7979}=-2$, and $\beta^*_{8460}=3$.
The parameters' values used in SSNAL and ADMM are set as the same as the test previously.
Besides, we also run both algorithms $20$ times randomly and draw the box plot for the estimated coefficients in Figure \ref{fig2}.
It can be seen clearly that the variables are selected correctly and their estimated values are almost accurate. The values of relative error and variance further validate the algorithm's high precision and stability. Therefore, this simple test demonstrates the strong performance of SSNAL in high-dimensional studies.

\begin{figure}[h]
\centering
\includegraphics[width=4.5in]{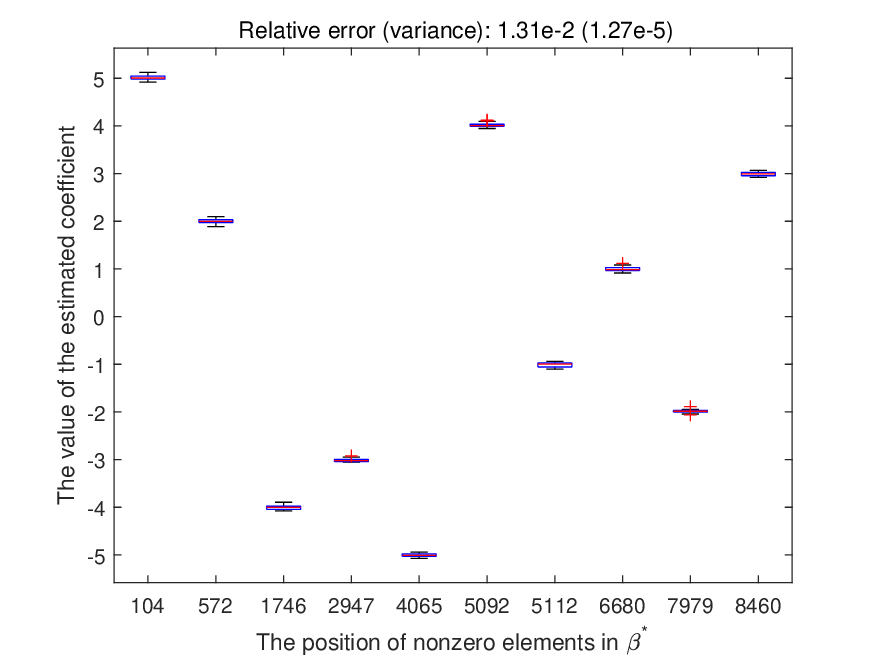}
\caption{{\small The calculation effect of SSNAL with $n=300$, $p=10000$ and $\lambda=0.1$.}}
\label{fig2}
\end{figure}

The second task is to illustrate the numerical advantages of SSNAL over ADMM in high-dimensional settings with adaptive lasso and lasso penalties.
In this test, we choose $n=500$ and $p=1000$. We set $a=5 \sqrt{2\log(p)/n}$ and $b=100 a$ to construct an interval such that
there are $20$ nonzero elements of underlying regression coefficient $\beta^*$ uniformly distributed in the interval $[a,b]$.
As in the previous test, we run SSNAL and ADMM $20$ times to estimate the regression coefficient $\beta^*$, and the positions of these non-zero components are assigned randomly at each time.
The average results regarding to ReErr, NNZ, Res, Time(s), and Iter are recorded in Table~\ref{tab:2}.
From the table, it is clear that SSNAL significantly outperforms ADMM in terms of computation time and iteration count, while achieving competitive accuracy.

\begin{table}[h]
\centering
\caption{Average results of SSNAL and ADMM with $n=500$ and $p=1000$.}
\label{tab:2}
\setlength{\tabcolsep}{1.5pt}
\begin{tabular}{llllll}
\hline\noalign{\smallskip}
Methods & ReErr & NNZ & Res & Time(s) & Iter\\
\noalign{\smallskip}\hline\noalign{\smallskip}
SSNAL$_a$ & 7.03e-4 (1.10e-4)& 20 (0)& 2.60e-7 (2.24e-7) & 0.06 (0.04) & 1.7 (0.67)  \\
\specialrule{0em}{3pt}{3pt}
SSNAL$_l$  & 2.90e-3 (6.15e-4)& 19.9 (0.31)& 1.34e-7 (2.93e-7) & 0.53 (0.20) & 4.9 (0.31)\\
\specialrule{0em}{3pt}{3pt}
ADMM$_a$ & 7.03e-4 (1.10e-4) & 20 (0) & 9.97e-7 (2.49e-9) & 52.37 (3.03)& 955 (39.79)\\
\specialrule{0em}{3pt}{3pt}
ADMM$_l$  & 2.90e-3 (6.00e-4) & 20 (0) & 9.94e-7 (3.56e-9) & 51.43 (3.26) & 944.7 (31.51) \\
\noalign{\smallskip}\hline
\end{tabular}
\end{table}

\subsection{Real Data Study}

In this section, we further evaluate the effectiveness of PLS and the numerical advantages of SSNAL with the workers' wage data which is available at \url{https://rdrr.io/cran/ISLR/man/Wage.html}.
This data set contains the wage information of $3000$ male workers in the Mid-Atlantic region, covering factors such as year, age, marriage status, race, education level, region, type of job, health level, health insurance information.
Specifically, in this test, we don't consider the data of year in which wage information was recorded and the logarithm of workers' wage for the sake of simplicity.
It should be noted that since the data only includes the Mid-Atlantic region, the use of regional indicators is redundant.
We note that there should be a non-linear relationship between the education level and the wage, so the education level is treated as the variable $T$ in the non-parameter part $g(T)$ in (\ref{sprm}).
In this test, we consider $6$ covariates to relate the wage for each worker, i.e., age, martial status, race, type of job, health level, and health insurance, which are denoted respectively as
$X_{ij}$ from $j=1$ to $6$ for each worker $i$.
More descriptions for each sample $X_i$ for index $j$ can be found in the second column of Table \ref{tab:3}.

For numerical convenience, we normalize all predictors with mean $0$ and variance $1$.
Moreover, for the adaptive lasso penalized models, the method for generating the weight vector $\omega$ is the same as the one in the low-dimensional case tested previously.
As before, we also use SSNAL and ADMM for problem (\ref{pls}) to estimate the coefficient $\beta$ in model (\ref{sprm}) with lasso and
adaptive lasso penalties, respectively.
The $\bar{\beta}$ estimated by SSNAL and ADMM with adaptive lasso penalty (named SSNALa and ADMMa) and lasso penalty (named SSNALl and ADMMl) are reported in the third-to-last column of Table \ref{tab:3}.
The numerical results of Res, NNZ, Iter, Time(s) are reported in bottom part of Table \ref{tab:3}.
These results indicate that adaptive lasso penalized model can select $5$ covariates, whereas the lasso penalized model cannot, and SSNAL requires fewer iterations and is significantly faster than ADMM.

\begin{table}[h]
\centering
\caption{Numerical results of SSNAL and ADMM  on the $3000$ male workers' wage.}
\label{tab:3}
\setlength{\tabcolsep}{1.5pt}
\begin{tabular}{clcccc}
\noalign{\smallskip}
\hline\noalign{\smallskip}
Variable & Description & $\bar{\beta}$(SSNALa) & $\bar{\beta}$(SSNALl) & $\bar{\beta}$(ADMMa) & $\bar{\beta}$(ADMMl)\\
\noalign{\smallskip}\hline\noalign{\smallskip}
$X_{i1}$ & Age of worker & 5.9035 & 5.6215 & 5.9035 & 5.6215 \\
\specialrule{0em}{3pt}{3pt}
$X_{i2}$  &Marital status:& 0 & 0.9409 & 0 & 0.9409\\
 & (1=Never Married, & & & & \\
 & 2=Married, & & & & \\
 & 3=Widowed, & & & & \\
  & 4=Divorced, & & & & \\
 & 5=Separated) & &  & & \\
\specialrule{0em}{3pt}{3pt}
$X_{i3}$ &Race: & 2.0303 & 2.1937 & 2.0303 & 2.1937\\
 & (1=Other, 2=Black, & & & & \\
 & 3=Asian, 4=White) & &    & & \\
\specialrule{0em}{3pt}{3pt}
$X_{i4}$ &Type of job: & 1.3296 & 1.6557  & 1.3296 & 1.6557 \\
 & (1=Industrial, & & & & \\
  & 2=Information) & & & & \\
\specialrule{0em}{3pt}{3pt}
$X_{i5}$ & Health level: & 3.4366 & 3.4914  & 3.4366 &  3.4914\\
 & (1=Good, 2=Very Good)   & &    & & \\
\specialrule{0em}{3pt}{3pt}
$X_{i6}$  & Health insurance: & 8.1894 & 8.1401 & 8.1894 & 8.1401\\
& (1=Yes, 0=No)  & &    & & \\
\noalign{\smallskip}\hline
Res & & 4.86e-7 & 1.71e-7 & 9.75e-7 &  6.29e-7\\
NNZ & & 5 & 6 & 5 &  6\\
Iter & & 3 & 2 & 277 & 16 \\
Time(s) & & 2.01 & 0.26 & 276.76 & 15.37 \\
\noalign{\smallskip}\hline
\noalign{\smallskip}
\end{tabular}
\end{table}
\section{Conclusions}\label{consec}

This paper focuses on a partially linear semiparametric regression model with an unknown regression coefficient and an unknown nonparametric function. Specifically, we proposed a PLS method to estimate and select the regression coefficient. We showed that the oracle property of the proposed PLS can be easily followed from some existing works. For practical implementation, this paper technically employed an efficient SSNAL method, which differs from almost all existing approaches by targeting the corresponding dual problem. Additionally, a semismooth Newton algorithm was used to solve the strongly semismooth nonlinear system involved in each iteration by making full use of the structure of lasso. Finally, we tested the algorithm and performed a performance comparison with ADMM with some random synthetic data and real data. The comparison results showed that PLS is very effective and the performance of SSNAL is highly efficient. This paper currently considers only the adaptive lasso penalty. The extensions to other convex or non-convex penalties are highly worthwhile for further research.
\section*{Acknowledgements}

The work of Peili Li is supported by the National Natural Science Foundation of China (Grant No. 12301420).
The work of Yunhai Xiao is supported by the National Natural Science Foundation of China (Grant No. 12471307 and 12271217), the National Natural Science Foundation of Henan Province (Grant No. 232300421018).
The work of Hanbing Zhu is supported by the National Natural Science Foundation of China (Grant No. 12201218).

\section*{Disclosure statement}

The authors report there are no competing interests to declare.


\begin{thebibliography}{9}

\bibitem{BYRD} R. H. Byrd, G. M. Chin, J. Nocedal, et al, A family of second-order methods for convex $\ell_1$-regularized optimization. \textit{Mathematical Programming}. \textbf{159} (2016), 435--467.

\bibitem{FH} F. H. Clarke, Optimization and nonsmooth analysis. \textit{John Wiley and Sons, New York.} (1983).

\bibitem{Engle} R. F. Engle, C. W. J. Granger, J. Rice, et al, Semiparametric estimates of the relation between weather and electricity sales. \textit{Journal of the American Statistical Association}. \textbf{81} (1986), 310--320.

\bibitem{FG1996} J. Q. Fan, Local polynomial modelling and its applications. \textit{Routledge.} \url{https://doi.org/10.1201/9780203748725}. (1996).

\bibitem{FanLi} J. Q. Fan and R. Z. Li, Variable selection via nonconcave penalized likelihood and its oracle properties. \textit{Journal of the American Statistical Association}. \textbf{96} (2001), 1348--1360.

\bibitem{long-data} J. Q. Fan and R. Z. Li, New estimation and model selection procedures for semiparametric modeling in longitudinal data analysis. \textit{Journal of the American Statistical Association}. \textbf{99} (2004), 710--723.

\bibitem{SEMP13} M. Fazel, T. K. Pong, D. F. Sun, et al, Hankel matrix rank minimization with applications to system identification and realization. \textit{SIAM Journal on Matrix Analysis and Applications}. \textbf{34} (2013), 946--977.

\bibitem{GM1976} D. Gabay and B. Mercier, A dual algorithm for the solution of nonlinear variational problems via finite element approximation. \textit{Computers $\&$ mathematics with applications}. \textbf{2} (1976), 17--40.

\bibitem{GR} R. Goebel and R. T. Rockafellar, Local strong convexity and local Lipschitz continuity of the gradient of convex functions. \textit{Journal of Convex Analysis}. \textbf{15} (2008), 263--270.

\bibitem{SMW} G. H. Golub and C. F. Van Loan, Matrix computations. \textit{Johns Hopkins University Press}. (1996).

\bibitem{BOOK1} W. H$\ddot{a}$rdle, H. Liang and J. Gao, Partially linear models. \textit{Springer Science $\&$ Business Media}. (2000).

\bibitem{JJL2015} Y. L. Jiao, B. T. Jin and X. L. Lu, A primal dual active set with continuation algorithm for the $\ell^0$-regularized optimization problem. \textit{Applied and Computational Harmonic Analysis}. \textbf{39} (2015), 400--426.

\bibitem{KK2007} S. Konishi and G. Kitagawa, Information criteria and statistical modeling. \textit{Springer Science $\&$ Business Media}. (2008).

\bibitem{K2021} E. Kwessi, Double penalized semi-parametric signed-rank regression with adaptive LASSO. \textit{Journal of Systems Science and Complexity}. \textbf{34} (2021), 381--401.

\bibitem{Pls} F. Li, Y. Q. Lu and G. R. Li, Variable selection for partially linear models via adaptive LASSO. \textit{Chinese Journal of Applied Probability and Statistics}. \textbf{28} (2012), 614--624.

\bibitem{LTX2008} G. R. Li, P. Tian and L. G. Xue, Generalized empirical likelihood inference in semiparametric regression model for longitudinal data. \textit{Acta Mathematica Sinica, English Series}. \textbf{24} (2008), 2029--2040.

\bibitem{LST} X. D. Li, D. F. Sun and K. C. Toh, A highly efficient semismooth Newton augmented Lagrangian method for solving Lasso problems. \textit{SIAM Journal on Optimization}. \textbf{28} (2018), 433--458.

\bibitem{error} H. Liang and R. Z. Li, Variable selection for partially linear models with measurement errors. \textit{Journal of the American Statistical Association}. \textbf{104} (2009), 234--248.

\bibitem{double} X. Ni, H. H. Zhang and D. W. Zhang, Automatic model selection for partially linear models. \textit{Journal of Multivariate Analysis}. \textbf{100} (2009), 2100--2111.

\bibitem{RR} R. T. Rockafellar, Convex analysis. \textit{Princeton University Press}. (1970).

\bibitem{Rockafellar1976a} R. T. Rockafellar, Augmented Lagrangians and applications of the proximal point algorithm in convex programming. \textit{Mathematics of Operations Research}. \textbf{1} (1976), 97--116.

\bibitem{Rockafellar1976b} R. T. Rockafellar, Monotone operators and the proximal point algorithm. \textit{SIAM Journal on Control and Optimization}. \textbf{14} (1976), 877--898.

\bibitem{RW} R. T. Rockafellar and R. J. B. Wets, Variational analysis. \textit{Grundlehren der mathematischen Wissenschaften}. (1998).

\bibitem{SPECKMAN} P. Speckman, Kernel smoothing in partial linear models. \textit{Journal
of the Royal Statistical Society. Series B (Methodological)}. \textbf{50} (1988), 413--436.

\bibitem{Stigler} S. M. Stigler, Gauss and the invention of least squares. \textit{The Annals of Statistics}. \textbf{9} (1981), 465--474.

\bibitem{LASSO} R. Tibshirani, Regression shrinkage and selection via the lasso. \textit{Journal of the Royal Statistical Society. Series B (Methodological)}. \textbf{58} (1996), 267--288.

\bibitem{FLASSO} R. Tibshirani, M. Saunders, S. Rosset, et al, Sparsity and smoothness via the fused lasso. \textit{Journal of the Royal Statistical Society. Series B (Methodological)}. \textbf{67} (2005), 91--108.

\bibitem{WLT2007} H. S. Wang, R. Z. Li and C. L. Tsai, Tuning parameter selectors for the smoothly clipped absolute deviation method. \textit{Biometrika}. \textbf{94} (2007), 553--568.

\bibitem{WKL2013} L. Wang, Y. D. Kim and R. Z. Li, Calibrating non-convex penalized regression in ultra-high dimension. \textit{The Annals of Statistics}. \textbf{41} (2013), 2505--2536.

\bibitem{WJ2003} Q. H. Wang and B. Y. Jing, Empirical likelihood for partial linear models. \textit{Annals of the Institute of Statistical Mathematics}. \textbf{55} (2003), 585--595.

\bibitem{spline} H. L. Xie and J. Huang, SCAD-penalized regression in high-dimensional partially linear models. \textit{The Annals of Statistics}. \textbf{37} (2009), 673--696.

\bibitem{Zou} H. Zou, The adaptive Lasso and its oracle properties. \textit{Journal of the American Statistical Association}. \textbf{101} (2006), 1418--1429.

\end{thebibliography}
\end{document}